%% file: main.tex
\pdfoutput=1
\documentclass[a4paper,cleveref,USenglish]{lipics-v2021}
\nolinenumbers

\newif\ifpublish
\hideLIPIcs  

\input{config/packages}

\input{config/macros}

\makeatletter
\let\@authorsaddresses\@empty
\makeatother

\title{Byzantine Consensus in the Random Asynchronous Model}

\author{George Danezis}{Mysten Labs and University College London}{}{}{}
\author{Jovan Komatovic}{EPFL}{}{}{}
\author{Lefteris Kokoris-Kogias}{Mysten Labs}{}{}{}
\author{Alberto Sonnino}{Mysten Labs and University College London}{}{}{}
\author{Igor Zablotchi}{Mysten Labs}{}{}{}

\authorrunning{G. Danezis, J. Komatovic, L. Kokoris-Kogias, A. Sonnino, I. Zablotchi}



\begin{document}

\maketitle
\begin{abstract}
  \input{sections/abstract}

\end{abstract}

\clearpage\newpage
\setcounter{page}{1}
\input{sections/intro}
\input{sections/model}
\input{sections/3f+1}
\input{sections/resilience}
\input{sections/negative}
\input{sections/other-models}
\input{sections/related}

\input{sections/conclusion}

\clearpage\newpage
\appendix
\section*{APPENDIX}
\input{sections/challenges}
\input{sections/proofs}
\input{sections/crash}

\bibliographystyle{plainurl}
\bibliography{main}

\end{document}

%% file: config/packages.tex
\usepackage{algorithmicx}
\usepackage{algorithm} 
\usepackage{algpseudocode}
\usepackage{caption}
\usepackage{xspace}
\usepackage{booktabs}

\algnewcommand{\BlueComment}[1]{\textcolor{blue}{\hfill\(\triangleright\) #1}}
\newfloat{module}{htbp}{loa}
\floatname{module}{Module}

%% file: config/macros.tex
\newcommand{\model}{random asynchronous model\xspace}

\newcommand{\MoDeL}{Random Asynchronous Model\xspace}
\newcommand{\ra}{random asynchrony\xspace}
\newcommand{\Ra}{Random asynchrony}
\newcommand{\Cnf}{\mathcal{C}(n,f)\xspace}
\newcommand{\para}[1]{\smallskip\noindent\textsf{\textbf{#1.}}}
\newcommand{\protocol}{\mathcal{A}\xspace}
\newcommand{\whp}{w.h.p.\xspace}

\newcommand{\as}{a.s.\xspace}

%% file: sections/abstract.tex
We propose a novel relaxation of the classic asynchronous network model, called the \textit{\model}, which removes adversarial message scheduling while preserving unbounded message delays and Byzantine faults. Instead of an adversary dictating message order, delivery follows a random schedule. We analyze Byzantine consensus at different resilience thresholds ($n=3f+1$, $n=2f+1$, and $n=f+2$) and show that our relaxation allows consensus with probabilistic guarantees which are impossible in the standard asynchronous model or even the partially synchronous model. We complement these protocols with corresponding impossibility results, establishing the limits of consensus in the \model.

%% file: sections/intro.tex
\section{Introduction}\label{sec:intro}


Byzantine Fault-Tolerant (BFT) consensus is a fundamental primitive in distributed computing, serving as the backbone of many applications, including blockchains~\cite{sok-consensus}. Solving BFT consensus in an asynchronous network model~\cite{FLP} is a well-studied problem~\cite{book,dag-rider,ren2017,mahimahi}. Typically, algorithms in this model assume that the adversary controls (1) a subset of faulty processes and (2) message delays, particularly the \textit{message schedule}—the order in which messages are delivered.


However, asynchronous BFT consensus is constrained by restrictive lower bounds~\cite{FLP, byz-generals}. These lower bounds often stem from the adversary's ability to impose arbitrary message schedules. Requiring an algorithm to function under the worst possible scheduling scenario often renders tasks impossible due to a single, highly unlikely counterexample, that is, a pathological schedule that violates the algorithm's properties~\cite{FLP}.

Yet, due to performance considerations, practical implementations of consensus circumvent the limitations of asynchrony by assuming that the schedule is not adversarial. For example, we are aware of asynchronous consensus protocols running in production over commodity wide-area networks for extensive periods of time without a random coin implementation~\cite{narwhal}---so as to avoid the performance overhead of cryptographically-secure randomness---without loss of liveness. This motivates our search for a model to explain this empirical observation.

Another common relaxation that circumvents the limitations of asynchrony is to assume \textit{periods of synchrony}, leading to the widely used \textit{partially synchronous} model~\cite{dwork1988consensus}. However, protocols based on partial synchrony lose liveness outside of these periods, which may occur due to poor network conditions or denial-of-service (DoS) attacks~\cite{consensus-dos}.

In this paper, in order to provide a sound basis for non-adversarial scheduling as assumed by some modern practical systems, and to avoid the limitations of partial synchrony, we propose an alternative relaxation of the asynchronous model: we study asynchronous networks without adversarial scheduling. Specifically, we ask:
\textit{What becomes possible in a network where message delays are unbounded, but the message schedule is not adversarial?}

To explore this question, we introduce a new variant of the asynchronous model, which we call the \textit{\model}. In this model, message delays remain unbounded, and the adversary still controls a subset of processes. However, the message schedule is no longer adversarial; instead, messages are delivered in a \textit{random} order.  
Our aim is not to argue that this model is a more realistic description of real networks; rather, we view it as a useful abstraction, inspired by assumptions made in practice~\cite{narwhal}, that lets us analyze the non-adversarial schedule assumption and explore which tasks it enables under unbounded delays while preserving Byzantine faults.
This work offers a new perspective on the role of adversarial scheduling in the standard asynchronous setting: by removing the adversary’s scheduling power, we explore what becomes achievable and what remains fundamentally impossible.
Our work is a part of the broader line of research that explores relaxations of the fully adversarial scheduler: from the fair-scheduler abstraction of Bracha and Toueg~\cite{BrachaT85}, through Aspnes's noisy scheduling for shared memory~\cite{Aspnes02}, to recent DAG-based BFT systems such as Tusk~\cite{narwhal} and Mahi-Mahi~\cite{mahimahi}.


To isolate the impact of random scheduling, we consider only deterministic algorithms, meaning processes do not have access to local randomness. We analyze the impact of our model on Byzantine consensus at different resilience thresholds: $n = 3f + 1$, $n = 2f + 1$, and $n = f + 2$, where $n$ is the total number of processes, and up to $f$ processes may be faulty.

We find that our model enables consensus protocols that are impossible under standard asynchrony. The key insight is that the \model prevents an adversary from blocking honest parties from communicating indefinitely. In traditional asynchrony, an adversary can delay messages indefinitely to prevent termination (e.g., in the FLP impossibility result~\cite{FLP}). In contrast, in the \model, the adversary cannot control the schedule, ensuring that honest parties can exchange messages within a bounded number of steps with high probability.

This last property is reminiscent of what synchrony guarantees, i.e., that any message sent by a correct process is received within a time bound $\Delta$. This leads to the natural question: \textit{How powerful is the \model with respect to the standard models---asynchrony, partial synchrony, and synchrony?} We answer this question in \Cref{sec:other-models}. In brief, we show that the \model sits between asynchrony and synchrony in terms of task solvability power, yet perhaps surprisingly, is \textit{incomparable} to partial synchrony. 
To illustrate the latter separation, consider the following two tasks: (1) deterministic consensus when at least one process may fail by crashing, and (2) Byzantine consensus for $n < 3f+1$. Task (1) is not solvable in the \model, as we show in \Cref{sec:negative}, but is solvable in partial synchrony~\cite{book}; conversely, task (2) is not solvable in partial synchrony~\cite{dwork1988consensus} yet becomes attainable in the \model with probabilistic safety (\Cref{sec:2f+1}).
Finally, we also show that the \model behaves like a variant of synchrony where timing guarantees hold with high probability 
instead of deterministically, thus confirming the intuitive similarity between the two models.

\subsection{Modeling Challenges}

Designing an asynchronous model with a non-adversarial schedule presented several challenges. Our initial attempt at a round-based model, while mathematically tractable, proved too rigid, enforcing communication patterns that were overly restrictive. We then explored probabilistic scheduling over entire message schedules, but this approach was unintuitive and impractical for analysis.
A more promising approach was to randomly select individual messages for delivery. However, this exposed a critical vulnerability: Byzantine processes could manipulate the scheduling distribution by flooding the system with messages, effectively regaining control over the schedule. Attempts to mitigate this by encrypting messages failed, as the adversary could still infer crucial protocol information from the message patterns and delivery timings. More details about our initial attempts to model random asynchrony can be found in \Cref{sec:model-challenge}.

These challenges led us to our final model: instead of selecting individual messages, the scheduler randomly selects sender-receiver pairs. At each step, a sender-receiver pair $(s, r)$ is chosen, and the earliest pending message from $s$ to $r$ is delivered. This approach prevents Byzantine nodes from biasing the schedule since the probability of a message being delivered depends only on the number of available sender-receiver pairs, not on the volume of messages sent by a single process.

A natural concern is whether removing adversarial scheduling trivializes the consensus problem. We address this in two ways. First, Appendix~\ref{sec:algo-challenge} presents a naive algorithm that fails to solve consensus when $n \leq 2f + 1$, highlighting a fundamental challenge: even though the \model guarantees eventual communication, Byzantine nodes can still equivocate, preventing honest processes from distinguishing between correct and faulty messages. Second, we complement our positive results with corresponding impossibility results, establishing close bounds on what is solvable in the \model.

\subsection{Our Results}

Our key contribution is the introduction of the \model, a novel relaxation of the classic asynchronous model that removes adversarial scheduling while preserving unbounded message delays and Byzantine faults. This relaxation allows us to achieve new bounds that were previously impossible under standard asynchrony.
We then design BFT consensus protocols in this new model, analyzing their feasibility at different resilience thresholds. Finally, we provide both positive results (demonstrating feasibility) and impossibility results (establishing the model's limits), summarized in \Cref{tab:ra_guarantees}.

\begin{table}[t]
  \centering
  \caption{Guarantees achievable (top) and impossible (bottom) in the \model.  “Strong” and “weak” validity are defined in \Cref{sec:model}.}
  \label{tab:ra_guarantees}
  \begin{tabular}{@{}llccc@{}}
    \toprule
    \multicolumn{2}{c}{$n$ versus $f$} & \textbf{Validity} & \textbf{Agreement} & \textbf{Termination} \\
    \midrule
    \multicolumn{5}{c}{\textbf{Positive results}}\\
    \cmidrule(lr){1-5}
        & $n=3f+1$ & Deterministic strong & Deterministic & Almost surely \\
        & $n=2f+1$ & With high probability strong & With high probability & Deterministic \\
        & $n=f+2$  & Deterministic weak  & With high probability & Deterministic \\
    \midrule
    \multicolumn{5}{c}{\textbf{Negative results}}\\
    \cmidrule(lr){1-5}
        & $n=3f+1$ & Deterministic strong & Deterministic & Deterministic \\
        & $n=2f+1$ & Almost surely strong & Almost surely & Deterministic \\
        & $n=f+2$  & With high probability strong & With high probability & Deterministic \\
    \bottomrule
  \end{tabular}
\end{table}


%% file: sections/model.tex
\section{Model \& Preliminaries}\label{sec:model}

\para{Processes} We consider a set $\Pi$ of $n$ processes, up to $f$ of which may be faulty. We consider the Byzantine-fault model, in which faulty processes may depart arbitrarily from the protocol. Throughout the paper, we assume that correct (non-faulty) processes have deterministic logic, i.e., they do not have access to a source of randomness.

\para{Cryptography} We make assumptions that are standard for Byzantine fault-tolerant algorithms: processes communicate through authenticated channels (e.g., using message authentication codes (MACs)) and have access to digital signatures whose properties cannot be broken by the adversary.

\para{Network} The processes communicate by sending messages over a fully-connected reliable network: every pair of processes $p$ and $q$ communicate over a link that satisfies the \textit{integrity} and \textit{no-loss} properties. Integrity requires that a message $m$ from $p$ be received by $q$ at most once and only if $m$ was previously sent by $p$ to $q$. No-loss requires that a message $m$ sent from $p$ to $q$ be eventually received by $q$.

\para{Random asynchrony} Processes send messages by submitting them to the network; a \textit{random scheduler} decides in which order submitted messages are delivered. We assume the scheduler delivers messages one by one, i.e., no two delivery events occur at exactly the same time. For convenience, we use a global discrete notion of time to reflect the sequence of delivery events: time proceeds in discrete steps $0, 1, \ldots$; each time step corresponds to a message delivery event in the system. Processes do not have access to this notion of time (they do not have clocks). Sending a message and performing local computation occur instantaneously, between time steps. Note that this notion of time does not bound message delays: an arbitrary amount of real time can pass between time steps.

We next describe the random scheduler. At any time $t$, let $P(t) \subseteq \Pi^2$ be the set of pairs of processes $(p,q)$ such that $p$ has at least one \textit{pending message} to $q$. We say that a message $m$ from $p$ to $q$ is pending at time $t$ if $p$ sends $m$ before $t$ and $q$ has not received $m$ by time $t$. At each time step $t$, the random scheduler draws a pair $(p,q)$ at random from $P(t)$ and delivers to $q$ the earliest message from $p$. We assume that there exists a constant $\Cnf > 0$ such that, for any $p$ and $q$, the probability that $(p,q)\in P(t)$ is drawn at time $t$ is at least $\Cnf$. In general, $\Cnf$ may depend on $n$ and $f$, but not on $t$. We assume that scheduling draws do not depend on the content of the message being delivered.

Furthermore, we assume that scheduling draws are conditionally independent of each other, given the history of messages sent and delivered. More formally, let \(\{X_t\}_{t\ge 0}\in P(t)\) be the random variables that represent the scheduling draws at each step $t$, and let \(\mathcal F_{t}\) be the history (all messages sent/delivered) up to step $t$. Then our lower bound assumption above can be written as 
\begin{equation}\label{ass:lb}
  \Pr[X_t = (p,q)\mid \mathcal F_{t}] \;\ge\; C(n,f)\quad\text{for any}\quad t\geq0
\end{equation}
and the conditional independence assumption can be written as
\begin{equation}\label{ass:independence}
    \Pr[X_{t_1},\dots,X_{t_k}\mid \mathcal F_{t_k}]
   = \prod_{i=1}^{k}\Pr[X_{t_i}\mid\mathcal F_{t_i}]
  \quad\text{for any } t_1<\dots<t_k.
\end{equation}
Intuitively, the scheduling draws are independent given the message history because the scheduler flips a fresh coin at every step.

\para{Consensus} Our algorithms solve two variants of binary Byzantine consensus: strong  and weak. Strong Byzantine consensus is defined by the following properties:
\begin{description}
    \item[Strong Validity] If all correct processes propose $v$, then correct processes that decide, decide~$v$. 
    \item[Agreement] If correct processes $p$ and $q$ decide $v$ and $w$ respectively, then $v = w$.
    \item[Termination] Every correct process decides some value.
\end{description}

Weak Byzantine consensus~\cite{weak-byz} has the same agreement and termination properties as the strong variant above, but has a different validity property:
\begin{description}
    \item[Weak Validity] If all processes are correct and propose $v$, then processes that decide, decide~$v$.
\end{description}

\para{Deterministic and probabilistic properties}
The probability of a schedule is the probability of the intersection of all its scheduling steps. 
We say that an algorithm $\mathcal{A}$ ensures a property $\mathcal{P}$ \textit{deterministically} if $\mathcal{P}$ holds in every execution of $\mathcal{A}$. Note that any algorithm that ensures a property $\mathcal{P}$ deterministically in the standard asynchronous model, must also ensure $\mathcal{P}$ deterministically in the \model.
Given an algorithm $\mathcal{A}$ and a property $\mathcal{P}$, we say that a schedule $S$ is \textit{bad for $\mathcal{P}$} if there exists an execution $E$ of $\mathcal{A}$ with schedule $S$ such that $\mathcal{P}$ does not hold in $E$.
An algorithm $\mathcal{A}$ ensures a property $\mathcal{P}$ \textit{almost surely (\as)} if the total probability of all schedules that are bad for $\mathcal{P}$ is $0$. 
An algorithm $\mathcal{A}$ ensures a property $\mathcal{P}$ \textit{with high probability (\whp)} if the total probability of all schedules that are bad for $\mathcal{P}$ is negligible; in this paper, a probability is negligible if it approaches $0$ exponentially with some parameter of the algorithm, for any (fixed) $n$ and $f$ (e.g., the number of communication steps). 

%% file: sections/3f+1.tex
\section{$n=3f+1$: Deterministic Safety, Termination Almost Surely}
\label{sec:warmup}

In this section we solve binary Byzantine consensus with deterministic strong validity and agreement, and termination almost surely. 

Our algorithm, shown in \Cref{alg:skeleton}, proceeds in rounds: in each round, processes attempt to decide using a \textsc{Round} procedure; if unsuccessful, processes update their estimate and try again in the next round. The \textsc{Round} procedure is similar to an adopt-commit object~\cite{ac1, ac2}: processes \textit{propose} a value and return a pair $(g,v)$, where $v$ is a value and $g$ is a grade, which can be either \textsc{Commit} or \textsc{Adopt}. If $g=\textsc{Commit}$, then it is guaranteed that all correct processes return the same value $v$ (possibly with different grades). If all correct processes propose the same value $v$, then all correct processes are guaranteed to return $v$ with a \textsc{Commit} grade.

\begin{algorithm}
\caption{Main consensus algorithm for $n=3f+1$: pseudocode at process $i$}
\label{alg:skeleton}
\begin{algorithmic} [1]
\State \textbf{procedure }\textsc{Propose}$(v_i)$: 
\BlueComment{$v_i \in \{0,1\}$}
\State \hskip1em $est_i \gets v_i$
\State \hskip1em $r_i \gets 0$
\State \hskip1em \textbf{while} \textsf{true}:
\State \hskip2em $(g,v) \gets$ \textsc{Round}$(r_i, est_i)$
\State \hskip2em \textbf{if} $g = \textsc{Commit}$: \textsf{decide}$(v)$\label{line:skeleton-decide}
\BlueComment{Only once}
\State \hskip2em $est_i \gets v$
\State \hskip2em $r_i \gets r_i + 1$
\end{algorithmic}
\end{algorithm}

\begin{algorithm}
\caption{Byzantine \textsc{Round} implementation for $n=3f+1$: pseudocode at process $i$}
\label{alg:aac-byz}
\begin{algorithmic} [1]
\State \textbf{procedure} \textsc{Round}$(r,v)$:
\State \hskip1em BRB-Broadcast $\langle \textsc{Init}, r, v \rangle$\label{line:fac-byz-init}
\BlueComment{Phase 1}
\State \hskip1em Wait to BRB-Deliver $\langle \textsc{Init}, r,\_ \rangle$ from $n-f$ processes\label{line:fac-byz-wait-init}
\State \hskip1em $\mathcal{H} \gets$ delivered $\langle \textsc{Init}, r,\_ \rangle$ messages 
\State \hskip1em $proposal \gets$ majority value in $\mathcal{H}$\label{line:fac-byz-adopt-phase1}
\State \hskip1em BRB-Broadcast $\langle \textsc{Echo}, r, proposal, \mathcal{H} \rangle$ to all processes
\BlueComment{Phase 2}\label{line:fac-byz-send-echo}
\State \hskip1em Wait to BRB-Deliver valid $\langle \textsc{Echo}, r, \_, \_ \rangle$ messages from $n-f$ processes\label{line:fac-byz-wait-echo}
\State \hskip1em \textbf{if} $\exists\, v^* \ne \bot$ such that I received $\geq 2f+1$ $\langle \textsc{Echo}, r, v^*, \_ \rangle$ messages:
\State \hskip2em \textbf{return} $\langle \textsc{Commit}, v^* \rangle$\label{line:fac-byz-commit}
\State \hskip1em \textbf{else}:
\State \hskip2em $v^* \gets$ majority value among received $\langle \textsc{Echo}, r, \_ , \_\rangle$ messages
\State \hskip2em \textbf{return} $\langle \textsc{Adopt}, v^* \rangle$\label{line:fac-byz-adopt}
\end{algorithmic}
\end{algorithm}

The core of the algorithm is the \textsc{Round} procedure, shown in \Cref{alg:aac-byz}. There are two types of messages: \textsc{Init} and \textsc{Echo}. Processes rely on Byzantine Reliable Broadcast (BRB)~\cite{book} for communication.\footnote{Any asynchronous BRB protocol is also correct in the \model.} Furthermore, all messages are signed. The algorithm has two phases. In phase 1, each correct process $p$ broadcasts (using BRB) its input value $v$ in an \textsc{Init} message (line~\ref{line:fac-byz-init}), and waits to deliver $n-f = 2f+1$ \textsc{Init} messages (line~\ref{line:fac-byz-wait-init}). Process $p$ adopts as its \textit{proposal} for phase 2, the majority value among the received \textsc{Init} messages (line~\ref{line:fac-byz-adopt-phase1}). Then, in phase 2, $p$ broadcasts an \textsc{Echo} message with $p$'s phase 2 proposal, as well as the $n-f$ (signed) \textsc{Init} messages that justify $v$ to be the majority phase 1 value (line~\ref{line:fac-byz-send-echo}); $p$ waits for $n-f$ valid \textsc{Echo} messages (line~\ref{line:fac-byz-wait-echo}). If all delivered \textsc{Echo} messages are for the same value $v^*$, then $p$ commits $v^*$ (line~\ref{line:fac-byz-commit}); otherwise $p$ adopts the majority value (line~\ref{line:fac-byz-adopt}).

We prove the correctness of our algorithm in \Cref{sec:proofs-det}. The main intuition is that at each round, a favorable schedule can help all correct processes agree (i.e., adopt the same estimate), by causing them to select the same majority value in phase 1 (e.g., by ensuring that they deliver \textsc{Init} messages from the same set of processes). Since the schedule is random, it has a non-zero chance of being favorable at each round. Thus, almost surely, the schedule will eventually be favorable. This is similar in spirit to random coin based Byzantine Randomized Consensus algorithms~\cite{book}, in which correct processes choose their next estimate by tossing a coin, if they do not manage to reach agreement in a given round. Here we are relying on the random schedule instead of the random coin. However, as we show later, in our model, consensus is solvable for settings where it is impossible for standard asynchrony even when equipped with a common coin.

\para{Crash-fault tolerant consensus} A similar approach can be used to solve crash-fault tolerant consensus with $n=2f+1$, with the same guarantees of deterministic safety, as well as termination almost surely. Our algorithm uses the same round-based structure in \Cref{alg:skeleton} and a modified \textsc{Round} procedure, that does not require reliable broadcast, and employs standard point-to-point messages instead. The main intuition is similar to our Byzantine algorithm: the random scheduler eliminates the need for a common coin by ensuring that correct processes eventually deliver messages in a favorable order which leads them to agree and thus terminate. \Cref{sec:crash-fault} gives our algorithm and proves its correctness.

%% file: sections/resilience.tex
\section{$n = 2f+1$: Deterministic Termination, Safety with High Probability}\label{sec:2f+1}

In this section, we pose $n=2f+1$ and solve strong Byzantine consensus such that safety (validity and agreement) holds with high probability and termination is deterministic.

\Cref{alg:resilient-consensus} shows our proposed protocol. The main idea is as follows: There are $f+1$ phases, each consisting of $R$ communication rounds, where $R$ is a parameter. $R$ is large enough so that correct processes hear from each other at least once within $R$ rounds with high probability. In each round, a correct process sends its set of \textit{accepted values} ($V_i$ in \Cref{alg:resilient-consensus}) to all other processes. A correct process $i$ \textit{accepts} a value from process $j$ at phase $p$ ($p = 1,\ldots,f+1$) if (1) $j$ is the origin of $v$ (i.e., the first signature on $v$ is by $j$), (2) $i$ has not already accepted a value from $j$, and (3) $v$ has valid signatures from $p$ distinct processes. We say that $i$ accepts a value $v$ at phase $p$ if $p$ is the earliest phase at which $i$ accepts $v$ (from any process).

After the $f+1$ phases are over, a correct process decides on the majority value within its set of accepted values (i.e., the value that appears most often). Note that at the end of the communication phases, each correct process must have at least one accepted value, since correct processes sign and send their input value in phase $1$ and the network is reliable.


The main intuition behind this protocol is that, if a correct process accepts a value $v$, then all correct processes will accept $v$ by the end of the execution, and thus all correct processes will have the same set of accepted values. Therefore, all correct processes can use a deterministic rule to decide the same value.

\begin{algorithm}
\caption{Binary Byzantine consensus for $n = 2f+1$; pseudocode for process $i$}
\label{alg:resilient-consensus}
\begin{algorithmic} [1]
\State \textbf{Local variables:}
\State \hskip1em $V_i = \varnothing$, a map from processes to signed values
\bigskip
\State \textbf{procedure} \textsc{Propose}$(v)$:
\State \hskip1em $V_i[i] \gets \textsc{Sign}(v)$
\State \hskip1em \textbf{for} $phase$ in $1\ldots f+1$:
\State \hskip2em \textbf{for} $round$ in $1\ldots R$:
\State \hskip3em Send $(V_i, phase, round)$ to all processes
\State \hskip3em $valid \gets 0$
\State \hskip3em \textbf{while} $valid < n-f$:
\State \hskip4em Receive a $(V_j, p, r)$ message\BlueComment{receive single msg per process, phase and round}
\State \hskip4em \textbf{if} $\exists v \in V_j:$ $v$ is signed by $phase$ distinct processes \textbf{and} $V_i[\textsc{Origin}(v)] = \varnothing$
\State \hskip5em $V_i[\textsc{Origin}(v)] \gets \textsc{Sign}(v)$
\State \hskip4em \textbf{if} $(p,r) = (phase, round)$
\State \hskip5em $valid \gets valid + 1$

\State \hskip1em \textbf{decide} $\textsc{MajorityValue}(V_i)$
\end{algorithmic}
\end{algorithm}

We now prove that \Cref{alg:resilient-consensus} satisfies the consensus properties in \Cref{sec:model}.

\begin{lemma}\label{lem:all-hear}
    With high probability, every correct process receives at least one message from every other correct process in each phase.
\end{lemma}
\begin{proof}
We start by fixing two correct processes $p$ and $q$ and upper bounding the probability that $q$ does not receive any message from $p$ after $R$ iterations. The probability of not delivering $p$'s round-$1$ message to $q$ is at most $(1-\Cnf)$ at each time step. There must be at least $R(n-f)$ time steps for $q$ to complete $R$ iterations, since $q$ waits for at least $n-f$ messages at each iteration, and each message received consumes one time step.
Thus, using assumptions~(\ref{ass:lb}) and~(\ref{ass:independence}), the probability of $q$ not observing $p$'s message is at most 
$$(1-\Cnf)^{R(n-f)} \approx e^{-R\Cnf(n-f)}.$$

To finish the proof, we upper-bound the probability of any correct process not observing the input value of some other correct process. We first compute the expected number of (ordered) pairs of processes $(p,q)$ such that $q$ does not observe the input value of $p$ at the end of the $R$ iterations. There are $n(n-1)$ possible pairs, so this expected value is at most $E = n(n-1)e^{-R\Cnf(n-f)}$. Now, by Markov's inequality, we have that 
$$\Pr(\text{at least one unreachable pair}) \leq E = n(n-1)e^{-R\Cnf(n-f)}.$$
For fixed $n$ and $f$, this probability approaches $0$ exponentially in $R$.
\end{proof}

\begin{lemma}\label{lem:all-accept}
    With high probability, every correct process accepts the input values of every other correct process.
\end{lemma}
\begin{proof}
    By \Cref{lem:all-hear}, every correct process receives at least one message from every other correct process in the first phase, \whp Since messages from correct processes always contain their input values with a valid signature, every correct processes $i$ will accept the input value of another correct process $j$ when $i$ receives the first message from $j$. 
\end{proof}

\begin{theorem}
    With $n=2f+1$, \Cref{alg:resilient-consensus} satisfies strong validity \whp
\end{theorem}
\begin{proof}
    Assume that all correct processes propose the same value $v$. Then, by \Cref{lem:all-accept}, all correct processes will accept at least $n-f=f+1$ copies of value $v$ \whp Since $f+1$ is a majority out of a maximum of $n=2f+1$ accepted values, correct processes decide $v$ \whp
\end{proof}

\begin{lemma}\label{lem:accept-early}
    If a value $v$ is accepted by a correct process $i$ at phase $p \leq f$, then all correct processes accept $v$ by the end of phase $p+1$, \whp
\end{lemma}
\begin{proof}
    If $i$ accepts $v$ at phase $p$, then $v$ must have signatures from at least $p$ processes. Process $i$ is not one of the $p$ processes, otherwise $i$ would have accepted $v$ at an earlier phase. Since $i$ accepts $v$, $i$ adds its signature to $v$ and will send $v$, as part of $V_i$, to all processes in phase $p+1$. By \Cref{lem:all-hear}, all correct processes will thus receive $v$ by the end of phase $p+1$ \whp, and will accept $v$ (if they haven't already), since $v$ has the required number of signatures.
\end{proof}

\begin{lemma}\label{lem:accept-last}
    If a value $v$ is accepted by a correct process $i$ at phase $f+1$, then all correct processes accept $v$ by the end of phase $f+1$, \whp
\end{lemma}
\begin{proof}
    If $i$ accepts $v$ at phase $f+1$, then $v$ must have signatures from at least $f+1$ processes; $i$ is not among these processes, otherwise $i$ would have accepted $v$ at an earlier phase. Since there are at most $f$ faulty processes, $v$ must have at least one signature from a correct process $j \ne i$.

    So $j$ must have accepted $v$ at an earlier phase $p \leq f$ and therefore, by \Cref{lem:accept-early}, all correct processes will accept $v$ by the end of phase $f+1$ \whp
\end{proof}

\begin{theorem}\label{thm:2f+1-agreement}
    With $n=2f+1$, \Cref{alg:resilient-consensus} satisfies agreement \whp
\end{theorem}
\begin{proof}
    By \Cref{lem:accept-early} and \Cref{lem:accept-last}, with high probability, correct processes have the same set of accepted values by the end of phase $f+1$, and thus decide the same value.
\end{proof}

\begin{theorem}
    With $n=2f+1$, \Cref{alg:resilient-consensus} satisfies deterministic termination.
\end{theorem}
\begin{proof}
    Follows immediately from the algorithm: correct processes only execute for $R(f+1)$ rounds. In each round, a correct process waits to receive $n-f$ messages from that round, which is guaranteed to occur in a finite number of steps, since there are at least $n-f$ correct processes and the network is reliable (no-loss property).
\end{proof}


\section{$n=f+2$: Deterministic Termination and Validity, Agreement \whp}

Interestingly, for $n=f+2$, we can solve weak Byzantine consensus with deterministic validity and termination, and agreement \textit{\whp}, using the same protocol in \Cref{alg:resilient-consensus}.

Weak validity is clearly preserved: if all processes are correct and have the same input value $v$, no other value is received by any process, and thus all processes decide $v$.

\begin{theorem}
    With $n=f+2$, \Cref{alg:resilient-consensus} satisfies deterministic weak validity.
\end{theorem}
\begin{proof}
    If all processes are correct and propose the same value $v$, then all $(V_i, phase, round)$ messages will have $v$ as their value, so no process can decide any other value.
\end{proof}

\begin{theorem}
    With $n=f+2$, \Cref{alg:resilient-consensus} satisfies agreement \whp
\end{theorem}
\begin{proof}
    Lemmas~\ref{lem:all-hear}, \ref{lem:all-accept}, \ref{lem:accept-early}, and \ref{lem:accept-last} still hold: their proofs are also valid if $n=f+2$. Thus the proof of this theorem is the same as the proof of \Cref{thm:2f+1-agreement}: By \Cref{lem:accept-early} and \Cref{lem:accept-last}, with high probability, correct processes have the same set of accepted values by the end of phase $f+1$, and thus decide the same value.
\end{proof}

\begin{theorem}
    With $n=f+2$, \Cref{alg:resilient-consensus} satisfies deterministic termination.
\end{theorem}
\begin{proof}
    Correct processes only execute for $R(f+1)$ rounds. In each phase and round, a correct process waits to receive $n-f$ valid messages from that phase and round. This wait is guaranteed to terminate since there are at least $n-f$ correct processes, correct processes can always produce a valid message (a message $(V_i, p, r)$ is valid if $p$ and $r$ are equal to the current phase and round, respectively), and the network is reliable (no-loss property).
\end{proof}

%% file: sections/negative.tex
\section{Negative Results}\label{sec:negative}
 
In this section we provide negative results that closely match our positive results from previous sections. Intuitively, we show that in the \model it is not possible to obtain Byzantine consensus protocols with more powerful guarantees than the protocols we propose in this paper. This shows that the \model, while avoiding some restrictions and impossibilities of the standard asynchronous model, is not overly permissive. Concretely, we prove the following three results. 

\begin{theorem}\label{thm:flp}
   No protocol can solve consensus in the \model with deterministic strong validity, agreement, and termination, if at least one process may fail by crashing.
\end{theorem}

\begin{theorem}\label{thm:2f+1}
    With $n=2f+1$, no protocol for Byzantine consensus in the \model can ensure strong validity and agreement with probability $1$, as well as deterministic termination. 
\end{theorem}

\begin{theorem}\label{thm:f+2}
    With $n=f+2, f\geq 2$, no protocol for Byzantine consensus in the \model can ensure strong validity and agreement with high probability, while ensuring deterministic termination.
\end{theorem}

\begin{proof}[Proof sketch for \Cref{thm:flp}]
    This result is equivalent to the FLP impossibility~\cite{FLP} in the \model, and the FLP proof holds in our model as well. Essentially, if at least one process can fail by crashing, there exists an infinite bivalent execution (the same execution as constructed in the FLP proof), which prevents processes from deciding without breaking agreement.
\end{proof}

Instead of proving \Cref{thm:2f+1} directly, we can prove the following stronger result, which also implies \Cref{thm:2f+1}.
\begin{theorem}\label{thm:flp-p1}
    No protocol for consensus in the \model can ensure strong validity and agreement almost surely, as well as deterministic termination, if at least one process may fail by crashing.
\end{theorem}
\begin{proof}
    Assume such a protocol $\protocol$ exists. Since $\protocol$ ensures deterministic termination, it must terminate in finite time in every execution. Furthermore, $\protocol$ ensures strong validity and agreement almost surely, so the probability of schedules in which $\protocol$ terminates in a bivalent state (in the sense of FLP~\cite{FLP}) must be $0$. Yet, following the FLP proof, there exists an infinite bivalent execution $E$.  It follows that $\protocol$ must decide in a finite prefix $\pi$ of $E$, when the state is still bivalent. Since $\pi$ has finite length, its schedule has non-zero probability. We have shown that $\protocol$ does not ensure agreement and validity almost surely, a contradiction.
\end{proof}

\begin{proof}[Proof of \Cref{thm:f+2}]
    Assume towards a contradiction that such a protocol $\protocol$ exists. Call an execution $E$ of $\protocol$ \textit{heterogeneous} if in $E$ at least 2 processes propose $0$ and at least 2 processes propose $1$. Assume \textit{wlog} that $p_1$ and $p_2$ propose $0$, while $p_{n-1}$ and $p_n$ propose $1$.

    Let $E$ be a heterogeneous execution in which all processes are correct. Let $S$ be the schedule of $E$. Let $\pi$ be shortest prefix of $S$ in which all processes have decided; in other words, we truncate any steps in $S$ after the processes have decided. The prefix $\pi$ must be finite since $\protocol$ ensures deterministic termination.

    We now describe three executions $E_1$, $E_2$, and $E_3$, with the same schedule $\pi$, and show that $\protocol$ must break either strong validity or agreement in one of the executions. This is sufficient to show that $\pi$ is bad (in the sense of \Cref{sec:model}) for strong validity or agreement.

    Let $E_1$ be an execution with schedule $\pi$, in which all processes behave identically to $E$; in $E_1$, $p_1$ and $p_2$ are correct, while $p_3,\ldots,p_n$ are Byzantine but behave correctly. By strong validity, $p_1$ must decide $0$ in $E_1$.

    Let $E_2$ be an execution with schedule $\pi$, in which again all processes behave identically to $E$; this time, $p_1$ and $p_n$ are correct, while $p_2,\ldots,p_{n-1}$ are Byzantine. Since $E_1$ and $E_2$ are indistinguishable to $p_1$, and $p_1$ has deterministic logic, $p_1$ must decide the same value in both executions, namely $0$. Thus, in order to satisfy agreement, $p_n$ must also decide $0$ in $E_2$. 

    Let $E_3$ be an execution with schedule $\pi$, in which again all processes behave as in $E$; this time, $p_{n-1}$ and $p_n$ are correct, while $p_1,\ldots,p_{n-2}$ are Byzantine. Since $E_2$ and $E_3$ are indistinguishable to $p_n$, and $p_n$ has deterministic logic, $p_n$ must decide the same value in both executions, namely $0$. But this breaks strong validity, as both correct processes ($p_{n-1}$ and $p_n$) have proposed $1$ in $E_3$.

    We have shown that $\pi$ is bad for agreement or strong validity. Since the length of $\pi$ is fixed (and finite) for fixed $n$ and $f$, the probability of $\pi$ is non-negligible according to our definition in \Cref{sec:model}. Therefore, $\protocol$ does not ensure strong validity and agreement with high probability, a contradiction.
\end{proof}

%% file: sections/other-models.tex
\section{The \MoDeL vs Standard Models}\label{sec:other-models}

In this section we compare the \model with the standard network models---asynchrony, synchrony, and partial synchrony---in terms of task solvability. We show that it is
\textbf{(i) strictly stronger} than full asynchrony (i.e., the set of tasks solvable in  \ra is a strict superset of the set of tasks solvable in asynchrony),
\textbf{(ii) strictly weaker} than synchrony, and
\textbf{(iii) incomparable} with partial synchrony, 
with respect to task solvability.  
Table~\ref{tab:comparisons} summarizes the relationships.

\begin{table}[h]
  \centering
  \caption{Relative power of the \model.}
  \begin{tabular}{lccc}
    \toprule
                     & Asynchrony & Partial synchrony & Synchrony\\
    \midrule
    \Ra & strictly stronger & incomparable & strictly weaker\\
    \bottomrule
  \end{tabular}
  \label{tab:comparisons}
\end{table}




\para{Asynchrony}
Every task solvable in asynchrony is also solvable in \ra because \emph{every finite prefix of an asynchronous schedule occurs with non-zero probability} under our scheduler.\footnote{Formally, let $S$ be any (possibly infinite) asynchronous schedule and $\pi$ a finite prefix of $S$. Since each delivery step is chosen independently with lower-bound probability $C(n,f)\!>\!0$, the probability of $\pi$ is at least $C(n,f)^{|\pi|}$.}
Conversely, the task ``Byzantine binary consensus with $n=2f+1$, deterministic termination, and safety holding with high probability’’ is solvable in the \model
(Section~\ref{sec:2f+1}) yet impossible in pure asynchrony by the standard network partition argument~\cite{Ben-Or83,BrachaT85}. Hence the \model is strictly stronger.

\para{Synchrony}
The \model is strictly more restrictive than synchrony: any task that is solvable in the \model is also solvable in synchrony, and there are tasks which are solvable in synchrony but not in the \model. As an example of the latter, consider deterministic crash-fault tolerant consensus: this task is impossible in the \model (\Cref{thm:flp}) but it is solvable in synchrony.

We now prove the former direction: any task solvable in the \model is also solvable in synchrony. Take any task $\tau$ and let $\protocol_{RA}$ be a protocol that solves $\tau$ in the \model. We can simulate $\protocol_{RA}$ in synchrony using an adapter $\Sigma_p$ for each process $p$ between $\protocol_{RA}$ and the network, in the following way:
\begin{enumerate}
    \item At each synchronous time step $t=0,1,\ldots$, each correct process $p$ executes all outstanding local steps of $\protocol_{RA}$, including those triggered by messages delivered since the previous time step. Process $p$ sends the same messages it would in $\protocol_{RA}$, by handing them to $\Sigma_p$. 
    \item The synchronous network delivers all messages send at time $t$ by time $t+1$. $\Sigma_p$ buffers any messages received by $p$ between $t$ and $t+1$.
    \item At time $t+1$, $\Sigma_p$ draws a fresh random permutation of the links on which $p$ has received messages, then delivers to $\protocol_{RA}$ the \emph{earliest} pending message on each link in that order.
\end{enumerate}

The resulting schedule visible to $\protocol_{RA}$ is valid in the \model, as any causal relation is preserved. Since the $\Sigma_p$ adapters collectively deliver messages from at most $n^2$ links at each synchronous time step, each link has a probability of at least $1/n^2 := \Cnf$ of being selected at each time step internal to the random asynchronous protocol $\protocol_{RA}$. Thus, the $\Sigma_p$ adapters preserve the probabilistic delivery properties of the \model. We have described a construction that simulates any random asynchronous algorithm in synchrony, thus showing that any task that is solvable in the \model is solvable in synchrony. 

\para{Partial synchrony}
The \model is incomparable with respect to partial synchrony in terms of solvability: there are tasks that are solvable in the \model, but not in partial synchrony, and vice-versa.

As an example of a task that is solvable in the \model but not in partial synchrony, recall the same $n=2f+1$ Byzantine consensus task above. \Cref{sec:2f+1} shows that this task is solvable in the \model. However, this task is not solvable in partial synchrony~\cite{dwork1988consensus}.

To show that some tasks are solvable in partial synchrony but not in the \model, consider the task of (deterministic) crash-fault tolerant consensus. This task is solvable in partial synchrony if $n\geq2f+1$~\cite{book}, but is not solvable in the \model if at least one process can crash (\Cref{thm:flp}). 

\para{Discussion} Although the random–asynchronous model is \emph{strictly} weaker than full synchrony in terms of task solvability, it can approximate synchronous behavior with high probability.  After \(R\) rounds of all-to-all communication, every pair of correct processes has exchanged at least one message with probability \(1-\exp(-\Omega(R))\) (\Cref{lem:all-hear}).  This mirrors the deterministic guarantee offered by a single synchronous round---modulo the known time bound~\(\Delta\).

Leveraging this property, we can implement, in the \model, a crash-fault detector $\sim\!\mathcal{P}$ that is \emph{perfect with high probability}:
\begin{description}
    \item[Strong completeness] Every process that crashes is eventually permanently suspected by every correct process.
    \item[Strong accuracy \whp] No correct process is suspected by any correct process \whp
\end{description}

The implementation is a classic heartbeat scheme. 
Each process broadcasts a heartbeat message once per round; process~\(p\) suspects~\(q\) iff it has not received a heartbeat from~\(q\) in the last \(R\) rounds. As we increase $R$, the probability of false positive suspicion decreases exponentially, through an argument similar to \Cref{lem:all-hear}. 

Such a failure detector is similar to the perfect $\mathcal{P}$ and eventually perfect $\diamond\mathcal{P}$ failure detectors, which can be implemented in synchrony and partial synchrony respectively~\cite{book}. 
Whether protocols that rely on \(\mathcal{P}\) or \(\Diamond\mathcal{P}\)
can be systematically translated to use
$\sim\!\mathcal{P}$---preserving their guarantees \whp---remains an open question. 

Intuitively, such a systematic translation seems unlikely to exist, as downgrading deterministic communication guarantees to hold \whp could break arbitrary internal algorithm invariants. 
Consider a synchronous algorithm that runs for $T$ rounds; such an algorithm is able to rely on the fact that it will receive messages from the same set of at least $n-f$ correct processes in each of the $T$ rounds---e.g., at the end of round $T$ each process $p$ may output the $(n-f)^{\text{th}}$ highest process id $p$ heard from in every one of the $T$ rounds. However, if we directly translate this algorithm to random asynchrony, by replacing each synchronous round with $R$ asynchronous rounds (for at total of $RT$ rounds), the guarantee of receiving messages from a common subset of $n-f$ processes holds merely with high probability, not deterministically. In the rare executions where $p$ hears from a common subset of fewer than \(n-f\) processes, the translated algorithm would try to access data that does not exist, violating its invariants. Any translation from synchrony to \ra would therefore need manual, algorithm-specific fallback logic to cope with those low-probability outliers, making a generic translation unlikely.


%

%% file: sections/related.tex
\section{Related Work}

\para{Random scheduling}
Similar assumptions to the \model have been explored by previous work. In message passing, Bracha and Toueg~\cite{BrachaT85} define the fair scheduler, which ensures that in each message round, there is a non-zero constant probability that every correct process receives messages from the same set of correct processes. Under this scheduler, they proposed deterministic asynchronous binary consensus protocols for crash and Byzantine fault models. 
More recently, Tusk~\cite{narwhal} and Mahi-Mahi~\cite{mahimahi} employ a form of random scheduling. Their random scheduler can be seen as a special case of ours: they only consider the $n=3f+1$ case and assume a standard round-based model with the subset of processes that a process ``hears from'' in a given round chosen \textit{uniformly} at random among all possibilities. They leverage the random scheduler to increase the probability to commit at each round, and thus to reduce latency, whereas our paper focuses on circumventing impossibility results in standard asynchrony, at different ratios of fault-tolerance. They conduct experiments without randomization on a wide-area network, without observing loss of liveness, which can serve as motivation for our work.

In shared memory, Aspnes~\cite{Aspnes02} defines noisy scheduling, in which the adversary may choose the schedule, but that adversarial schedule is perturbed randomly. Under this assumption, deterministic asynchronous consensus becomes achievable. Also in shared memory, previous work introduce a \textit{stochastic scheduler}~\cite{AlistarhSV15,AlistarhCS16}, which schedules shared memory steps randomly. This line of work shows that many lock-free algorithms are essentially wait-free when run against a stochastic scheduler, because the schedules that would break wait-freedom have negligible probability.

\para{Randomized consensus}
A large body of research leverages random coins to circumvent the FLP impossibility result~\cite{FLP}, which states that deterministic consensus is impossible in crash-prone asynchronous systems.
In this approach, protocols relax deterministic termination to probabilistic termination, assuming processes have access to a source of (cryptographically-secure) randomness that cannot be predicted by the adversary. Randomized consensus protocols employ either local coins~\cite{Ben-Or83, ritas, waterbear}---which produce randomness independently and locally at each process, without coordination with other processes---or common coins~\cite{Rabin83, oracles-constantinople, signature-free, narwhal, mahimahi, bullshark}---which, through the use of coordination and strong cryptographic primitives, ensure that all correct processes receive the same random output with some probability.

Our algorithms for the $n=3f+1$ setting resemble existing coin-based random consensus protocols, with the randomness moved from the process logic to the schedule. In fact, all of the coin-based consensus protocols we examined can be transformed, with minor changes, into deterministic (coin-less) algorithms in the \model. A natural question, then, is whether our model is equivalent to the standard asynchronous model with coin tosses. It is not: in the standard model, achieving safety \textit{\whp} when $n < 3f$ is impossible, even if processes have access to randomness. This is because of a standard network partition argument: the adversary can partition correct processes into two sets that never exchange messages, allowing Byzantine to force different decisions in each set. By contrast, our model makes long-lived network partitions occur with negligible probability, allowing safety \textit{\whp} even for $n<3f$.


\para{Probabilistic quorum systems}
A line of work on probabilistic quorum systems~\cite{prob-quorums, Yu06} relaxes quorum  intersection to be probabilistic rather than deterministic, and allows for probabilistic correctness guarantees. However, they are vulnerable to an adversarial scheduler~\cite{AiyerAB05}. ProBFT~\cite{probft} addresses this through the use of verifiable random functions for quorum selection. These works are similar to ours in that correctness is probabilistic rather than deterministic, but their approach focuses on the $n=3f+1$ setting, and is mainly aimed at scalability and efficiency (e.g., communication complexity), whereas we aim to circumvent impossibilities in standard asynchrony across various fault tolerance ratios.

%% file: sections/conclusion.tex
\section{Conclusion}

We introduce the random asynchronous model, a novel relaxation of the classic asynchronous model that replaces adversarial message scheduling with a randomized scheduler. By eliminating the adversary’s ability to indefinitely delay messages, our model circumvents traditional impossibility results in asynchronous Byzantine consensus while preserving unbounded message delays and tolerating Byzantine faults. Our approach avoids the need for synchronized periods (as in partial synchrony) or cryptographic randomness (as in randomized consensus), offering a foundation for practical alternatives to existing asynchronous systems. 
We demonstrated that this relaxation enables new feasibility results across different resilience thresholds: deterministic safety and probabilistic termination for $n=3f+1$, deterministic termination with safety holding with high probability (\whp) for $n=2f+1$, and weak validity with \whp agreement for $n=f+2$. These results are complemented by impossibility bounds, showing our protocols achieve near-optimal guarantees under the model. 

Future work could explore extensions of this model to other distributed computing problems, such as state machine replication, and investigate empirical performance trade-offs in real-world deployments. By bridging the gap between theoretical impossibility and practical assumptions, we believe our model opens avenues for more efficient and resilient distributed protocols.

%% file: sections/challenges.tex
\section{Challenges}\label{sec:challenges}

\subsection{Modeling Challenge}\label{sec:model-challenge}
Our aim is to propose a model for asynchrony without adversarial scheduling that is (1) general, i.e., does not restrict algorithm design (2) easy to work with for proofs, (3) usable by practical algorithms, and (4) intuitive.
In the course of defining the current model, we came up with several other possibilities that do not meet the aims above:
\begin{enumerate}
    \item A round-based model, similar to the fair scheduler model of Bracha and Toueg~\cite{BrachaT85}. In each communication round, a correct process sends a message to every process and waits to hear back from $n-f$ processes. The random scheduler assumption is: in each round, a correct process has a non-trivial (i.e., lower-bounded by a constant) probability of hearing from any subset of $n-f$ processes. This model has the advantage of being easy to work with, but is too restrictive, as it restricts algorithms to the round-based structure.
    \item A model which places probability directly on entire schedules, instead of on individual communication steps: each valid schedule has a non-trivial probability of occurring. We found this model to be un-intuitive and difficult to work with.
    \item A model in which the next message to be delivered is drawn, according to some distribution, from all currently pending messages (i.e., messages that have been sent but not yet delivered). The distribution must ensure that every message has a non-trivial probability to be scheduled next. This model is general (does not restrict algorithm structure), intuitive, and easy to work with, but has the following crucial flaw. Byzatine processes can skew the scheduling distribution by producing a large number of messages (potentially under the guise of retransmitting them as part of the reliable links assumption). If, at any given time, most pending messages are from Byzantine processes, then these messages are more likely to be delivered first, effectively reverting the model to an adversarial scheduler.
    \item Similarly to the previous proposal: at each scheduling step, a sender-receiver pair $(s,r)$ is drawn \textit{uniformly at random}, and the earliest pending message from $s$ to $r$ is the next message delivered in the system. This model is intuitive, and easy to work with, while also fixing the message injection attack by Byzantine processes: the number of pending messages from a (potentially Byzantine) process $s$ to process $r$ does not influence the probability distribution of $s$'s messages to be delivered before other messages. The only drawback is with respect to generality: the uniform distribution on the sender-receiver pairs is a strong assumption.
\end{enumerate}

Our final model is similar to the last proposal above, while solving the generality problem by removing the uniform distribution assumption. Instead, we simply assume that the probability of each sender-receiver pair being drawn is non-negligible.

\subsection{Algorithmic Challenge}\label{sec:algo-challenge}
Take the following naive (and incorrect) binary consensus algorithm for the $n\leq2f+1$ cases ($n=2f+1$ and/or $n=f+2$):
\begin{itemize}
    \item Round $0$: Processes initially sign and send their input value to all other processes, and wait for such messages from $n-f$ processes.
    \item Rounds $1$--$R$ (where $R$ is a parameter): Processes sign and send their entire history of sent and received messages to all processes and wait for valid such messages from $n-f$ processes.
    \item At the end of round $R$, correct processes decide on, say, the lowest input value they have received.
\end{itemize}

This algorithm is subject to the following attack:
\begin{itemize}
    \item Assume all correct processes have $1$ as their input value.
    \item The $f$ Byzantine processes do not send any messages to the $n-f$ correct processes up until and including round $R-1$. Otherwise, Byzantine processes act as correct processes whose input values are $0$, including accepting messages from correct processes.
    \item Thus, no correct process is aware of $0$ as a valid input value before round $R$.
    \item Let $p$ be some correct processes that the Byzantine processes have agreed on before the start of the execution. The attack attempts to cause $p$ to decide on a different value than the other correct processes, breaking agreement.
    \item In round $R$, Byzantine processes send correctly constructed messages to $p$, containing their entire communication histories. If the random scheduler delivers even one of these messages to $p$ before the end of the round, then $p$ becomes aware of the input value $0$ for the first time in round $R$ (and therefore does not have time to inform the other correct processes).
    \item After round $R$, $p$ must decide $0$ (being the lowest input value it is aware of), while the other processes must decide $1$ (not being aware of $0$ as a valid input value), breaking agreement.
\end{itemize}

This attack succeeds if the random scheduler delivers a message from a Byzantine process to $p$ in round $R$, among the first $n-f$ messages delivered to $p$ in that round. This has a non-trivial chance of occurring.

This algorithm illustrates the main challenge of designing correct algorithms under the \model when $n\leq 2f+1$: even if the model ensures that all correct processes communicate with each other eventually, Byzantine processes can still equivocate and correct processes do not necessarily know which messages are from correct processes and which are not.

%% file: sections/proofs.tex
\section{Proofs for Deterministic Safety Algorithms}\label{sec:proofs-det}

In this section we prove the correctness of our algorithm in \Cref{sec:warmup}.
We first prove that the \textsc{Round} procedure in \Cref{alg:aac-byz} satisfies the properties below, and then prove that \Cref{alg:skeleton} solves consensus under Byzantine faults.
\begin{description}
    \item[Strong Validity] If all correct processes propose the same value $v$ and a correct process returns a pair $\langle \textsc{Grade}, v' \rangle$, then $\textsc{Grade} = \textsc{Commit}$ and $v' = v$.
    \item[Consistency] If any correct process returns  $\langle \textsc{Commit}, v \rangle$, then no correct process returns $(\cdot, v' \ne v)$.
    \item[Termination] If all correct processes propose, then every correct process eventually returns.
\end{description}

In our proofs we rely on the following properties of Byzantine Reliable Broadcast (BRB)~\cite{book}:
\begin{description}
    \item[BRB-Validity] If a correct process $p$ broadcasts a message $m$, then every correct process eventually delivers $m$.
    \item[BRB-No-duplication] Every correct process delivers at most one message.
    \item[BRB-Integrity] If some correct process delivers a message $m$ with sender $p$ and process $p$ is correct, then $m$ was previously broadcast by $p$.
    \item[BRB-Consistency] If some correct process delivers a message $m$ and another correct process delivers a message $m'$, then $m = m'$.
    \item[BRB-Totality] If some message is delivered by any correct process, every correct process eventually delivers a message.
\end{description}

\begin{lemma}\label{lem:byz-round-validity}
    With Byzantine faults and $n=3f+1$, \Cref{alg:aac-byz} satisfies strong validity.
\end{lemma}
\begin{proof}
    If all correct processes propose the same value $v$, then at least $2f+1$ processes BRB-broadcast an \textsc{Init} message for $v$, and therefore at most $f$ processes BRB-broadcast an \textsc{Init} message for $1-v$. Thus $v$ will be the majority value among all \textsc{Init} messages delivered in phase 1, at all correct processes. Thus all correct processes will BRB-broadcast an \textsc{Echo} message for $v$. Furthermore, no Byzantine process can produce a valid \textsc{Echo} message for $1-v$, since to do so would require a set of $2f+1$ \textsc{Init} message with a majority value of $1-v$. This is impossible due to the properties of BRB and the fact that at most $f$ processes have BRB-broadcast an \textsc{Init} message for $1-v$.
    So, all valid \textsc{Echo} messages received by correct processes will be for $v$, so all correct processes will commit $v$ at line~\ref{line:fac-byz-commit}.
\end{proof}

\begin{lemma}
    With Byzantine faults and $n=3f+1$, \Cref{alg:aac-byz} satisfies consistency.
\end{lemma}
\begin{proof}
    If a correct process $p_1$ commits $v$ at line~\ref{line:fac-byz-commit}, then it must have delivered a set $S_1$ of $2f+1$ \textsc{Echo} messages for $v$ at line~\ref{line:fac-byz-wait-echo}. Take now another process $p_2$ and consider the set $S_2$ of $2f+1$ \textsc{Echo} messages it delivers at line~\ref{line:fac-byz-wait-echo}. By quorum intersection, $S_1$ and $S_2$ must intersect in at least $f+1$ messages. By the BRB-Consistency property, these $f+1$ messages must be identical at $p_1$ and $p_2$. Thus $p_2$ delivers at least $f+1$ \textsc{Echo} messages for $v$, which constitutes a majority of the $2f+1$ \textsc{Echo} messages it delivers overall. So if $p_2$ commits a value at line~\ref{line:fac-byz-commit}, then it must commit $v$, and if $p_2$ adopts a value at line~\ref{line:fac-byz-adopt}, then it must adopt $v$.
\end{proof}

\begin{lemma}
    With Byzantine faults and $n=3f+1$, \Cref{alg:aac-byz} satisfies termination.
\end{lemma}
\begin{proof}
    Follows immediately from the algorithm and from the properties of Byzantine Reliable Broadcast. Processes perform two phases; the only blocking step of each phase is waiting for $n-f$ messages (lines~\ref{line:fac-byz-wait-init} and~\ref{line:fac-byz-wait-echo}). This waiting eventually terminates, by the BRB-Validity property and the fact that there are at least $n-f$ correct processes.
\end{proof}

\begin{theorem}\label{thm:validity-byz}
    With Byzantine faults and $n=3f+1$, \Cref{alg:skeleton} satisfies strong validity.
\end{theorem}
\begin{proof}
    This follows from the strong validity property of the \textsc{Round} procedure (\Cref{lem:byz-round-validity}): if all correct processes propose $v$ to consensus, then all correct processes propose $v$ to \textsc{Round} in the first round, where by \Cref{lem:byz-round-validity}, all correct processes commit $v$, and thus all correct processes decide $v$ at line~\ref{line:skeleton-decide}.
\end{proof}

\begin{theorem}\label{thm:agreement-byz}
    With Byzantine faults and $n=3f+1$, \Cref{alg:skeleton} satisfies agreement.
\end{theorem}
\begin{proof}
    Let $r$ be the earliest round at which some process decides and let $p$ be a process that decides $v$ at round $r$. We will show that any other process $p'$ that decides, must decide $v$. 
    
    For $p$ to decide $v$ at round $r$, \textsc{Round} must output $(\textsc{Commit}, v)$ in that round. Thus, by the consistency property of \textsc{Round}, $\textsc{Round}(r,\cdot)$ must output $(\cdot, v)$ at all correct processes. If $\textsc{Round}(r,\cdot)$ outputs $(\textsc{Commit}, v)$ for $p'$, then $p'$ decides $v$ at round $r$ (line~\ref{line:skeleton-decide}). Otherwise, all correct processes input $v$ to $\textsc{Round}(r+1,\cdot)$, and by the strong validity property, all processes (including $p'$) will output $(\textsc{Commit}, v)$ and decide $v$ at round $r+1$.
\end{proof}

\begin{theorem}\label{thm:termination-byz}
    With Byzantine faults and $n=3f+1$, \Cref{alg:skeleton} satisfies termination.
\end{theorem}
\begin{proof}
     We can describe the execution of the protocol as a Markov chain with states $0,\ldots,n-f=2f+1$; the system is at state $i$ if $i$ correct processes have estimate ($est_i$ variable) equal to $0$ before invoking $\textsc{Round}$. Due to the strong validity property of the $\textsc{Round}$ procedure, states $0$ and $2f+1$ are absorbing states. There is a non-zero transition probability from each state (including $0$ and $2f+1$), to state $0$ or $2f+1$, or both (we show this below). Therefore, almost surely, the system will eventually reach one of the two absorbing states and remain there. Once this happens (i.e., once all processes have the same $est_i$ variable), the strong validity property of $\textsc{Round}$ ensures that all processes (who have not decided yet) will decide within a round.
    
    It only remains to show that there is a non-zero transition probability from each state to at least one of the absorbing states $0$ and $2f+1$. Consider a state $i \notin \{0,2f+1\}$; there is a schedule $S$ with non-zero probability which leads the system from $i$ to $0$ or $2f+1$ in one invocation of \textsc{Round}. We consider two cases:
    \begin{itemize}
        \item $i < f+1$: in this case $0$ is the minority value among correct processes. In schedule $S$, the $n-f$ \textsc{Init} messages delivered by correct processes at line~\ref{line:fac-byz-wait-init} are all from correct processes. Thus, every correct process sees $i$ $0$s and $2f+1-i$ $1$s; $1$ is the majority value, so all correct processes adopt it for phase 2. In phase 2, $S$ again ensures that the $n-f$ \textsc{Echo} messages delivered by correct processes at line~\ref{line:fac-byz-wait-echo} are all from correct processes. Thus, all correct processes see $2f+1$ \textsc{Echo} messages for $1$ and commit $1$, bringing the system to state $0$.
        \item $i \geq f+1$: in this case $0$ is the majority value among correct processes. This case is symmetrical with respect to the previous one: the only difference is that all correct processes adopt $0$ (the majority value) at the end of phase 1, and all correct processes deliver $2f+1$ \textsc{Echo} messages for $0$, thus committing $0$ and bringing the system to state $2f+1$.
    \end{itemize}
\end{proof}

%% file: sections/crash.tex
\section{Crash-fault Tolerant Consensus in the \MoDeL}\label{sec:crash-fault}

\subsection{Definition}
In the crash-fault model, faulty processes may permanently stop participating in the protocol at any time, but otherwise follow the protocol.
Crash-fault tolerant consensus is defined by the following properties:
\begin{description}
    \item[Validity] If a process decides $v$, then $v$ was proposed by some process. 
    \item[Uniform Agreement] If processes $p$ and $q$ decide $v$ and $w$ respectively, then $v = w$.
    \item[Termination] Every correct process decides some value.
\end{description}

\subsection{Algorithm}

Our algorithm for crash-fault tolerant consensus uses the same round-based structure, shown in \Cref{alg:skeleton}, as our Byzantine consensus algorithm from \Cref{sec:warmup}. We use a different implementation of the \textsc{Round} procedure, shown in \Cref{alg:aac-crash}.

\begin{algorithm}
\caption{Crash-tolerant \textsc{Round} implementation: pseudocode at process $i$}
\label{alg:aac-crash}
\begin{algorithmic} [1]
\State \textbf{procedure} \textsc{Round}$(r,v)$:
\State \hskip1em Send $\langle \textsc{Init}, r, v \rangle$ to all processes\label{line:fac-crash-init}
\BlueComment{Phase 1}
\State \hskip1em Wait for $n-f$ $\langle \textsc{Init}, r,\_ \rangle$ messages\label{line:fac-crash-wait-init}
\State \hskip1em \textbf{if} $\exists\, v^*$ such that I received $\geq f+1$ $\langle \textsc{Init}, r, v^* \rangle$ messages:\label{line:fac-crash-if-init}
\State \hskip2em $proposal \gets v^*$\label{line:fac-crash-phase1-v} 
\State \hskip1em \textbf{else}:
\State \hskip2em $proposal \gets \bot$ \label{line:fac-crash-phase1-bot} 
\State \hskip1em Send $\langle \textsc{Echo}, r, proposal \rangle$ to all processes
\BlueComment{Phase 2}\label{line:fac-crash-send-echo}
\State \hskip1em Wait for $n-f$ $\langle \textsc{Echo}, r, \_ \rangle$ messages\label{line:fac-crash-wait-echo}
\State \hskip1em \textbf{if} $\exists\, v^* \ne \bot$ such that I received $\geq f+1$ $\langle \textsc{Echo}, r, v^* \rangle$ messages:
\State \hskip2em \textbf{return} $\langle \textsc{Commit}, v^* \rangle$\label{line:fac-crash-commit}
\State \hskip1em \textbf{else if} $\exists\, v^* \ne \bot$ such that I received $\geq 1$ $\langle \textsc{Echo}, r, v^* \rangle$ messages:
\State \hskip2em \textbf{return} $\langle \textsc{Adopt}, v^* \rangle$\label{line:fac-crash-adopt-det}
\State \hskip1em \textbf{else}:
\State \hskip2em $v^* \gets$ value in first $\langle \textsc{Init}, r,\_ \rangle$ message received\label{line:fac-crash-adopt-rand}
\State \hskip2em \textbf{return} $\langle \textsc{Adopt}, v^* \rangle$ 
\end{algorithmic}
\end{algorithm}

The \textsc{Round} algorithm consists of two
phases. In the first phase, every correct process proposes a value by sending it to
all processes (line~\ref{line:fac-crash-init}). Then it waits to receive proposals from
a quorum of processes (line~\ref{line:fac-crash-wait-init}). If a process observes that all responses contain the same
phase-one proposal value then it proposes that value for the second phase (line~\ref{line:fac-crash-phase1-v}). If a process does not obtain a unanimous set of proposals in the first phase, the process simply proposes $\bot$ for the second phase (line~\ref{line:fac-crash-phase1-bot}).

Note that as a result of this procedure, if two processes propose a value different from $\bot$ for the second phase, they propose exactly the same value.
Let this value be called $v^*$.

The purpose of the second phase is to verify if $v^*$ was also observed by enough other processes.
After a process receives $n-f$ phase-two messages (line~\ref{line:fac-byz-wait-echo}), it checks if more than $f$ phase-two proposals are equal to $v^*$, and if so commits $v^*$ (line~\ref{line:fac-crash-commit}). A process adopts $v^*$ if it receives $v^*$ in the second phase, but is unable to collect enough $v^*$ values to decide (line~\ref{line:fac-crash-adopt-det}). Finally, it is possible that a process does not receive $v^*$ in the second phase (either because no such value was found in phase one or simply because it has received only $\bot$ in phase two); in this case the process adopts the fist value it received in phase one (line~\ref{line:fac-crash-adopt-rand}). 

As in the Byzantine case, the main intuition is that at each round, a favorable schedule can help processes agree (i.e., adopt the same estimate), by causing them to adopt the same estimate at line~\ref{line:fac-crash-adopt-rand} (e.g., by ensuring that the first \textsc{Init} message they receive is from the same process). Since the schedule is random, it has a non-zero chance of being favorable at each round. Thus, almost surely, the schedule will eventually be favorable. 

\subsection{Proofs}
We begin with a few lemmas which establish that the \textsc{Round} procedure in \Cref{alg:aac-crash} satisfies these properties:
\begin{description}
    \item[Integrity] If a process returns $(\cdot, v)$, then $v$ was proposed by some process.
    \item[Strong Validity] If all correct processes propose the same value $v$ and a process returns a pair $\langle \textsc{Grade}, v' \rangle$, then $\textsc{Grade} = \textsc{Commit}$ and $v' = v$.
    \item[Consistency] If any correct process returns  $\langle \textsc{Commit}, v \rangle$, then no process returns $(\cdot, v' \ne v)$.
    \item[Termination] If all correct processes propose, then every correct process eventually returns.
\end{description}

\begin{lemma}
    With crash faults and $n=2f+1$, \Cref{alg:aac-crash} satisfies integrity.
\end{lemma}
\begin{proof}
    We say that a value $v$ is \textit{valid} if it is the input value of some process. We want to show that processes only return valid values. We observe that (1) \textsc{Init} messages only contain valid values (line~\ref{line:fac-crash-init}), and therefore (2) \textsc{Echo} messages only contain valid values or $\bot$ (lines~\ref{line:fac-crash-if-init}--\ref{line:fac-crash-send-echo}). If a process $p$ returns $v$ at line~\ref{line:fac-crash-commit} or~\ref{line:fac-crash-adopt-det}, then $p$ received at least one \textsc{Echo} message for $v$, and thus $v$ is valid by (2) above. If $p$ returns $v$ at line~\ref{line:fac-crash-adopt-rand}, then $p$ received at least one \textsc{Init} message for $v$, and thus $v$ is valid by (1) above.
\end{proof}

\begin{lemma}
    With crash faults and $n=2f+1$, \Cref{alg:aac-crash} satisfies strong validity.
\end{lemma}
\begin{proof}
    If all processes propose the same value $v$, then all processes send $\langle\textsc{Init}, v\rangle$ at line~\ref{line:fac-crash-init}; all processes receive at least $n/2$ $\langle\textsc{Init}, v\rangle$ messages (since $n-f = f+1 \geq n/2$); all processes adopt $v$ as their proposal for the second phase and send $\langle\textsc{Echo}, v\rangle$ at line~\ref{line:fac-crash-send-echo}; all processes receive $n-f=f+1$ $\langle\textsc{Echo}, v\rangle$ and return $\langle\textsc{Commit}, v\rangle$.
\end{proof}

\begin{lemma}\label{lem:fac-no-two-echo}
    If a process $p$ sends $\langle\textsc{Echo}, v\rangle$ at line~\ref{line:fac-crash-send-echo}, then no process sends $\langle\textsc{Echo}, v'\rangle$, for any $v'\ne v$.
\end{lemma}
\begin{proof}
    Assume the lemma does not hold. Then $p$ must have received more than $n/2$ $\langle\textsc{Init}, v\rangle$ messages, and some process $p'$ must have received more than $n/2$ $\langle\textsc{Init}, v'\rangle$ for some $v'\ne v$. By quorum intersection, it follows that some process must have sent both an $\langle\textsc{Init}, v\rangle$ and $\langle\textsc{Init}, v'\rangle$ message. This is a contradiction, as processes only send one \textsc{Init} message at line~\ref{line:fac-crash-init}.
\end{proof}

\begin{lemma}
    With crash faults and $n=2f+1$, \Cref{alg:aac-crash} satisfies consistency.
\end{lemma}
\begin{proof}
    If process $p$ returns $\langle\textsc{Commit}, v\rangle$, it must do so at line~\ref{line:fac-crash-commit}, after having received $f+1$ $\langle\textsc{Echo}, v\rangle$ messages. Therefore, every other process $p'$ that returns, must receive at least one $\langle\textsc{Echo}, v\rangle$ message. If $p'$ also receives $f+1$ $\langle\textsc{Echo}, v\rangle$ messages, then it returns $\langle\textsc{Commit}, v\rangle$. Otherwise, by \Cref{lem:fac-no-two-echo}, $p'$ cannot receive an \textsc{Echo} message for any other value $v'\ne v$, so $p'$ returns $\langle\textsc{Adopt}, v\rangle$ at line~\ref{line:fac-crash-adopt-det}.
\end{proof}

\begin{lemma}
    With crash faults and $n=2f+1$, \Cref{alg:aac-crash} satisfies termination.
\end{lemma}
\begin{proof}
    This follows immediately from the construction of the algorithm. Processes perform two phases; the only blocking step of each phase is waiting for $n-f$ messages (lines~\ref{line:fac-crash-wait-init} and~\ref{line:fac-crash-wait-echo}). This waiting eventually terminates, since there are at least $n-f$ correct processes.
\end{proof}

Now we can show that the consensus protocol in \Cref{alg:skeleton} is correct under crash faults.

\begin{theorem}\label{thm:validity-crash}
    With crash faults, \Cref{alg:skeleton} satisfies validity.
\end{theorem}

To prove the theorem, we first prove the following lemma:
\begin{lemma}\label{lem:crash-validity}
    If a process proposes $v$ to the \textsc{Round} procedure, then $v$ is the consensus input of some process.
\end{lemma}
\begin{proof}
    We prove the result by induction on the round $r$. For the base case: If $r = 0$, then all processes input their consensus inputs into $\textsc{Round}$. For the induction case: assume the lemma holds up to $r > 0$, and consider the case of $r + 1$. By the induction hypothesis and the integrity property of the \textsc{Round} procedure, any output of $\textsc{Round}(r,\cdot)$ must be the consensus input of some process. Since the input of $\textsc{Round}(r+1,\cdot)$ is the output of $\textsc{Round}(r,\cdot)$, the lemma must also hold for the case of $r+1$. This completes the induction.
\end{proof}

\begin{proof}[Proof of \Cref{thm:validity-crash}]
     If a process $p$ decides a value $v$, then the $\textsc{Round}$ procedure must have output $(\textsc{Commit}, v)$ at line \ref{line:skeleton-decide}. By \Cref{lem:crash-validity} and the integrity property of the \textsc{Round} procedure, it follows that $v$ is the consensus input of some process.
\end{proof}

\begin{theorem}\label{thm:crash-agreement}
    With crash faults, \Cref{alg:skeleton} satisfies uniform agreement.
\end{theorem}
\begin{proof}
    Let $r$ be the earliest round at which some process decides and let $p$ be a process that decides $v$ at round $r$. We will show that any other process $p'$ that decides, must decide $v$. 
    
    For $p$ to decide $v$ at round $r$, \textsc{Round} must output $(\textsc{Commit}, v)$ in that round. Thus, by the consistency property of \textsc{Round}, $\textsc{Round}(r,\cdot)$ must output $(\cdot, v)$ at all processes. If $\textsc{Round}(r,\cdot)$ outputs $(\textsc{Commit}, v)$ for $p'$, then $p'$ decides $v$ at round $r$ (line~\ref{line:skeleton-decide}). Otherwise, no other value than $v$ can be input to $\textsc{Round}(r+1,\cdot)$, and by the strong validity property, all processes (including $p'$) will output $(\textsc{Commit}, v)$ and decide $v$ at round $r+1$.
\end{proof}

\begin{theorem}\label{lem:termination-crash}
    With crash faults, \Cref{alg:skeleton} satisfies termination almost surely.
\end{theorem}
\begin{proof}
    We can describe the execution of the protocol as a Markov chain with states $0,\ldots,2f+1$; the system is at state $i$ if $i$ processes have estimate ($est_i$ variable) equal to $0$ before invoking $\textsc{Round}$. Due to the strong validity property of the $\textsc{Round}$ procedure, states $0$ and $2f+1$ are absorbing states. There is a non-zero transition probability from each state, other than $0$ and $2f+1$, to each other state (we show this below). Therefore, almost surely, the system will eventually reach one of the two absorbing states and remain there. Once this happens (all processes have the same $est_i$ variable), the strong validity property of $\textsc{Round}$ ensures that all processes (who have not decided yet) will decide within a round.
    
    It only remains to show that there is a non-zero transition probability from each state, other than $0$ and $2f+1$, to each other state. Consider a state $i \notin \{0,2f+1\}$ and a state $j$; there is a schedule $S$ with non-zero probability which leads the system from $i$ to $j$ in one invocation of \textsc{Round}. In $S$, every process receives \textsc{Init} messages for both $0$ and $1$ at line~\ref{line:fac-crash-wait-init}, and thus all processes send \textsc{Echo} messages with $\bot$ at line~\ref{line:fac-crash-send-echo}. Thus, all processes return at line~\ref{line:fac-crash-adopt-rand}. In $S$, $j$ processes receive an \textsc{Init} message with value $0$ first, so they adopt $0$, and $2f+1-j$ processes receive an \textsc{Init} message with value $1$ first, so they adopt $1$, bringing the system to state $j$.

\end{proof}